\documentclass[runningheads,a4paper]{llncs}

\usepackage[margin=30mm]{geometry}
\usepackage[british,UKenglish,USenglish,american]{babel}
\usepackage[T1]{fontenc}               
\usepackage[utf8]{inputenc}        	
\usepackage{graphicx}           

\usepackage{rotating}
\usepackage{amsmath}
\usepackage{amssymb}
\usepackage{tikz}
\usepackage{url}
\usepackage{cite}

\usepackage{color}

\newcommand{\N}{\mathbb{N}}

\begin{document}
\title{Online Exploration of Rectangular Grids}

\author{Hans-Joachim B\"{o}ckenhauer\inst{1} \and Janosch Fuchs\inst{2} \and
  Ulla Karhum\"{a}ki\inst{3}\fnmsep\thanks{Part of this work has been done
    while the third author was visiting ETH Z\"{u}rich.} \and Walter
  Unger\inst{2}}
\institute{Department of Computer Science, ETH
  Z\"{u}rich, Switzerland\\\email{hjb@inf.ethz.ch} \and
  Department of Computer Science, RWTH Aachen University,
  Germany\\\email{fuchs@cs.rwth-aachen.de, walter.unger@cs.rwth-aachen.de} \and
  Department of Mathematics, University of Manchester, United
  Kingdom\\\email{ulla.karhumaki@postgrad.manchester.ac.uk}}
\authorrunning{H-J.~B\"{o}ckenhauer, J.~Fuchs, U.~Karhum\"{a}ki, W.~Unger}

\maketitle

\begin{abstract}
In this paper, we consider the problem of exploring unknown environments with autonomous agents. We model the environment as a graph with edge weights and analyze the task of visiting all vertices of the graph at least once. The hardness of this task heavily depends on the knowledge and the capabilities of the agent. In our model, the agent sees the whole graph in advance, but does not know the weights of the edges. As soon as it arrives in a vertex, it can see the weights of all the outgoing edges. We consider the special case of two different edge weights $1$ and $k$ and prove that the problem remains hard even in this case. We prove a lower bound of $11/9$ on the competitive ratio of any deterministic strategy for exploring a ladder graph and complement this result by a $4$-competitive algorithm. All of these results hold for undirected graphs. Exploring directed graphs, where the direction of the edges is not known beforehand, seems to be much harder. Here, we prove that a natural greedy strategy has a linear lower bound on the competitive ratio both in ladders and square grids.

\end{abstract}

\section{Introduction}

Making an autonomous agent explore some unknown environment is a well-studied and important task \cite{KP, PK,FT, monta, MMS}. Usually, the environment is modeled by a graph, and the agent works in discrete time steps, moving from one vertex to a neighboring one in one time unit. The goal is to minimize the number of time steps the agent needs to accomplish some task, e.\,g., visiting all the vertices at least once or finding some special destination vertex. Many different models have been considered regarding the information the agent has about the graph in the beginning or learns with its visit to a newly visited vertex, for an overview, we point the reader to the surveys of Ghosh and Klein \cite{GK} and Berman \cite{B}. The problem of exploring an unweighted graph by visiting all its edges was introduced by Deng and Papadimitriou \cite{DP}. Kalyanasundaram and Pruhs presented a slight adaptation of this scenario, in which all vertices of a weighted graph have to be visited \cite{KP}. In \cite{PK}, the authors considered online routing algorithms for finding paths between the vertices of plane graphs.

In this paper, we are using the following model. The graph is weighted and either directed or undirected. At the beginning, the agent sees the whole undirected graph but only the weights of the edges incident to the starting vertex $s_0$, and the directions of the edges going out from $s_0$. However, as soon as a vertex $v$ is visited, the agent learns the weights of all the edges incident to $v$ and the directions of the edges going out from $v$. The agent has unlimited memory and can hence record all information it has gained about the graph. Although this model seems to give quite a lot of power to the agent, one can easily see that the resulting problem still is very hard and no deterministic agent can reach a satisfactory solution in arbitrary graphs. Thus, we restrict our attention to rectangular grids, especially ladders and square grids. We are interested in parameterizing the problem with respect to the admissible set of edge weights. As a first step, in this paper, we restrict the edge weights to two different values $1$ and $k$, for some arbitrary natural number $k\ge 5$. It turns out that, even in this very restricted scenario, no deterministic strategy for the agent can guarantee to achieve an optimal or asymptotically optimal solution.

We prove that, for exploring undirected weighted ladder graphs, that is, grids of dimension $2\times n$, there exists an algorithm \textsc{alg} with competitive ratio $4$. We complement this result with a lower bound showing that no deterministic algorithm can reach a better competitive ratio than $\frac{11}{9}$. After that, we show that, for directed weighted $2 \times n$ grids, the problem becomes much harder, by showing a sample graph for which the competitive ratio of any deterministic algorithm tends to $\frac{n}{a}$, where $a$ is a constant. We also extend this result to $n\times n$ grids.

\section{Preliminaries}

Our exploration results in this paper are for different special case of \emph{grids}. Below we are giving formal definitions of \emph{graphs}, \emph{grids} and \emph{paths}.

\begin{definition}[Graph] An \emph{undirected graph} is a pair $G=(V,E)$, where $V$ is a finite set of \emph{vertices} and $E\subseteq \lbrace \{ x,y\} \mid x \neq y$ and$ \ x,y \in V \rbrace$ is the set of \emph{edges}. A \emph{directed graph} is a pair $G=(V,E)$, where $V$ is a finite set of vertices and $E\subseteq V \times V$ is a set of directed edges (also called \emph{arcs}). We only consider directed graphs without loops, i.\,e., without edges of the form $(v,v)$ for $v \in V$. In an \emph{edge-weighted graph}, a function $c: E\rightarrow \N$ assigns a weight to every edge. Note that, in a directed, edge-weighted graph, the edges $(v,w)$ and $(w,v)$ may have different weights.

\end{definition}

\begin{definition}[Grid] For $m,n \geq 1$, the \emph{undirected $m\times n$ grid} is the graph $G_{m,n}=(V,E),$ where $V=\lbrace (i,j) \mid 1\leq i \leq m, 1\leq j \leq n \rbrace$ and $E =\lbrace (i,j),(i+1,j) \mid 1\leq i \leq m-1, 1\leq j \leq n \rbrace \cup \lbrace (i,j),(i,j+1) \mid 1\leq i \leq m, 1\leq j \leq n-1 \rbrace $.

In this paper, we are mostly considering grids forming a ladder or a square. That is, $m=2$ or $m=n$, respectively.

\end{definition}

\begin{definition}[Path] Let $G=(V,E)$ be an undirected graph. A sequence
of vertices $U  =(u_0,u_1,\ldots,u_{\ell})$ with $u_i \in V$, for $0 \leq i \leq \ell$, is
called a \emph{path} from $u_0$ to $u_{\ell} $ if, for every pair $(u_i, u_{i+1})$ with $0 \leq i \leq n-1$, there is
an edge $\{ u_i, u_{i+1} \} \in E$. Note that a path might visit some vertices several times, that is, it is possible that $u_i=u_j$ if $i\neq j$. 

Let $p =(u_1,u_2,\ldots,u_{\ell})$ be a path and let $e = \{x,y\} \in E$. We say that $e \in P$ if there exists $1\leq i\leq \ell-1$ such that $x=u_i$ and $y= u_{i+1}$.

Paths in directed graphs can be defined analogously.
\end{definition}

For analyzing the quality of some algorithm controlling the behaviour of the agent, we use the so-called \emph{competitive analysis}, which is the standard tool for the analysis of online algorithms. To this end, we start with a formal definition of online problems and the competitive ratio. These are the standard definitions that can be found, for example, in \cite{DK,ST,BE}.

\begin{definition}[Online problem] An \emph{online problem} $\Pi$ consists of a class
of instances $\mathcal{I}$ and a class of solutions $\mathcal{O}$. Every instance $I \in \mathcal{I}$ is a sequence of requests $I = ( x_1,x_2,\ldots,x_n )$
and every output $O \in \mathcal{O}$ is a sequence of answers $O = (y_1,y_2,\ldots,y_n)$, where
$n \in \N^+$ (thus, all instances and solutions are finite).

An \emph{online algorithm} $\textsc{alg}$ computes
the output $\textsc{alg}(I) = (y_1,y_2,\ldots,y_n)$, where $y_i$ only depends on $x_1,x_2,\ldots,x_i$
and $y_1,y_2,\ldots,y_{i-1}$.
\end{definition}

\begin{definition}[Competitive ratio] Let $\Pi$ be an online minimization problem, let
\textsc{alg} be a consistent online algorithm for $\Pi$, and let \textsc{opt} be an optimal algorithm for $\Pi$. For $r \geq 1$, \textsc{alg} is \emph{$r$-competitive} for
$\Pi$ if there is a non-negative constant $\alpha$ such that, for every instance $I \in \mathcal{I}$,  $$\textsc{cost(alg)}\leq r \cdot\textsc{cost(opt)} + \alpha.$$  If this inequality holds with $\alpha = 0$, we
call \textsc{alg} \emph{strictly $r$-competitive}. 

\end{definition}

Note that graph exploration is not exactly an online problem according to the above definition because the presentation of the input, i.\,e., in our case, the order in which the edge weights are revealed, is not chosen freely by an adversary, but depends on the decisions made by the algorithm. Nevertheless, we can use the same framework as for online algorithms to analyze it.

We start with formally defining the variants of the graph exploration problem we consider in this paper.

\begin{definition}[Graph exploration problem] The \emph{graph exploration problem in a known undirected graph structure with unknown edge weights}, denoted by GEP-UW, is the following minimization problem. An undirected graph $G=(V,E)$ with a start vertex $s_0 \in V$ is given. An agent is located in $s_0$, and its goal is to find a path through the graph visiting each vertex at least once. In one time step, the agent can move from its current vertex to any of the neighboring vertices. Whenever the agent visits a vertex for the first time, the weights of all incident edges are revealed. The agent knows the graph structure from the beginning and has unlimited computing power, in particular, it can remember all edge weights seen so far. The goal is to minimize the sum of edge weights on the path traveled by the agent. More precisely, for some algorithm $\textsc{alg}$ using the agent on some graph $G$, we denote by $p_A=v_1,\ldots,v_\ell$ the path traveled by the agent, and by $\textsc{cost}_{\textsc{alg}}(G)= \sum_{i=1}^{\ell-1} c({v_i,v_{i+1}})$ the cost of $p_A$.

The \emph{graph exploration problem in a directed graph with known graph structure and unknown edge weights and orientations}, denoted by GEP-DW, is defined in an analogous way. The agent knows the underlying undirected graph beforehand. The difference to GEP-UW is that every edge has a unique orientation in which it may be traveled that becomes known with the first visit to an incident vertex, and then the weight of an edge is revealed only for the outgoing edges of a vertex. Note that we only consider directed graphs here for which, for any edge $\{u,v\}$ of the underlying undirected graph, exactly one of the arcs $(u,v)$ and $(v,u)$ is present.
\end{definition}

We consider any deterministic online algorithm approach for the decision making of the agent. That is, the algorithm sees, and remembers, the weights (and directions of the outgoing edges for the directed case) of every edge leaving from any visited vertex. We want to analyze deterministic online algorithms for GEP-UW and GEP-DW.

We say that an \emph{agent} uses the online algorithm. Hence, the agent sees the edges leaving from the current vertex, and recognizes the vertices at the end of those edges. We assume that the starting vertex of the agent's path is $s_0=(1,1)$. In the figures, it is located in the top left corner.

An optimal solution for visiting each vertex of a given graph is the corresponding solution computed by an exact offline algorithm. We denote an arbitrary, but fixed optimal path by $p_O=v_1,\ldots,v_k$. For any graph, we get the cheapest possible cost for visiting all the vertices of some instance $I \in \mathcal{I}$ by using an exact offline algorithm. We denote this cost by $$\textsc{cost}_{\textsc{opt}}(I)=\sum_{i=1}^{k-1} c({v_i,v_{i+1}}).$$ In this paper, we are measuring the quality of \textsc{alg} by comparing it to the optimal algorithm. The result of this comparison is the already presented competitive ratio that we denote by $\textsc{comp}_{\textsc{alg}}(I).$ By the definition of the competitive ratio, it is easy to see that, if, for some algorithm \textsc{alg} and for all instances $I$ of size $n$, $$\textsc{comp}_{\textsc{alg}}(I) \xrightarrow [n \rightarrow \infty]{}r,$$ then the algorithm \textsc{alg} is $r$-competitive. Moreover, if, for some algorithm \textsc{alg} and for all instances $I$ of size $n$, $$\textsc{comp}_{\textsc{alg}}(I) \xrightarrow [n \rightarrow \infty]{}1,$$ then we say that the algorithm \textsc{alg} always finds an \emph{almost optimal strategy} for visiting all the vertices of these instances.

\section{Undirected Graphs}

In this section, we prove a lower bound and an upper bound on the competitive ratio of deterministic online algorithms for GEP-UW on undirected weighted ladder graphs $G_{2,n}$. The motivation of studying these bounds for our model arises as, even though the agent has quite a lot of power, one can easily see that the resulting problem still is very hard and no deterministic agent can reach a satisfactory solution in arbitrary graphs. Thus, we restrict our attention to ladders. Because of the hardness of the problem, we restrict the edge weights to two different values $1$ and $k$, for some arbitrary natural number $k \geqslant 5$. At first we prove that, for GEP-UW on $G_{2,4}$, no deterministic algorithm can reach a better strict competitive ratio than $\frac{5}{4}$. Then we generalize this result and get a lower bound of $\frac{11}{9}$ on the (non-strict) competitive ratio for $G_{2,n}$. Finally, we complement the result with a 4-competitive algorithm.

\subsection{Lower Bounds}

In Theorem 1, we prove that, for any deterministic online algorithm \textsc{alg}, $\frac{5}{4}$ is a lower bound on the strict competitive ratio for GEP-UW for a known ladder graph $G_{2,4}$, in which the edges have weights 1 and $k$, for any $k \geqslant 5$. Next, in Theorem 2, we generalize this result by showing that $\frac{11}{9}$ is a lower bound of the competitive ratio for GEP-UW on a known ladder graph $G_{2,n}$, in which the edges have weights 1 and $k$, again for any $k \geqslant 5$.

\begin{theorem} Let $G_{2,4}$ be a known undirected and weighted ladder graph, where weights of the edges are 1 or $k$, for a constant $k\geqslant5$. The strict competitive ratio for any deterministic online algorithm \textsc{alg} for GEP-UW on $G_{2,4}$ has a lower bound of $\frac{5}{4}$.

\end{theorem}

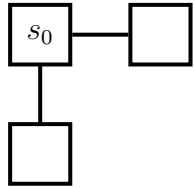
\begin{figure}
\begin{center}
\scalebox{1.0}{
\ifx\du\undefined
  \newlength{\du}
\fi
\setlength{\du}{15\unitlength}
\begin{tikzpicture}
\pgftransformxscale{1.000000}
\pgftransformyscale{-1.000000}
\definecolor{dialinecolor}{rgb}{0.000000, 0.000000, 0.000000}
\pgfsetstrokecolor{dialinecolor}
\definecolor{dialinecolor}{rgb}{1.000000, 1.000000, 1.000000}
\pgfsetfillcolor{dialinecolor}
\pgfsetlinewidth{0.100000\du}
\pgfsetdash{}{0pt}
\pgfsetdash{}{0pt}
\pgfsetmiterjoin
\definecolor{dialinecolor}{rgb}{1.000000, 1.000000, 1.000000}
\pgfsetfillcolor{dialinecolor}
\fill (8.407500\du,2.485000\du)--(8.407500\du,3.985000\du)--(9.907500\du,3.985000\du)--(9.907500\du,2.485000\du)--cycle;
\definecolor{dialinecolor}{rgb}{0.000000, 0.000000, 0.000000}
\pgfsetstrokecolor{dialinecolor}
\draw (8.407500\du,2.485000\du)--(8.407500\du,3.985000\du)--(9.907500\du,3.985000\du)--(9.907500\du,2.485000\du)--cycle;
\pgfsetlinewidth{0.100000\du}
\pgfsetdash{}{0pt}
\pgfsetdash{}{0pt}
\pgfsetmiterjoin
\definecolor{dialinecolor}{rgb}{1.000000, 1.000000, 1.000000}
\pgfsetfillcolor{dialinecolor}
\fill (5.407500\du,2.485000\du)--(5.407500\du,3.985000\du)--(6.907500\du,3.985000\du)--(6.907500\du,2.485000\du)--cycle;
\definecolor{dialinecolor}{rgb}{0.000000, 0.000000, 0.000000}
\pgfsetstrokecolor{dialinecolor}
\draw (5.407500\du,2.485000\du)--(5.407500\du,3.985000\du)--(6.907500\du,3.985000\du)--(6.907500\du,2.485000\du)--cycle;
\pgfsetlinewidth{0.100000\du}
\pgfsetdash{}{0pt}
\pgfsetdash{}{0pt}
\pgfsetmiterjoin
\definecolor{dialinecolor}{rgb}{1.000000, 1.000000, 1.000000}
\pgfsetfillcolor{dialinecolor}
\fill (5.407500\du,5.485000\du)--(5.407500\du,6.985000\du)--(6.907500\du,6.985000\du)--(6.907500\du,5.485000\du)--cycle;
\definecolor{dialinecolor}{rgb}{0.000000, 0.000000, 0.000000}
\pgfsetstrokecolor{dialinecolor}
\draw (5.407500\du,5.485000\du)--(5.407500\du,6.985000\du)--(6.907500\du,6.985000\du)--(6.907500\du,5.485000\du)--cycle;
\pgfsetlinewidth{0.100000\du}
\pgfsetdash{}{0pt}
\pgfsetdash{}{0pt}
\pgfsetbuttcap
{
\definecolor{dialinecolor}{rgb}{0.000000, 0.000000, 0.000000}
\pgfsetfillcolor{dialinecolor}
\definecolor{dialinecolor}{rgb}{0.000000, 0.000000, 0.000000}
\pgfsetstrokecolor{dialinecolor}
\draw (8.357329\du,3.235000\du)--(6.957671\du,3.235000\du);
}
\pgfsetlinewidth{0.100000\du}
\pgfsetdash{}{0pt}
\pgfsetdash{}{0pt}
\pgfsetbuttcap
{
\definecolor{dialinecolor}{rgb}{0.000000, 0.000000, 0.000000}
\pgfsetfillcolor{dialinecolor}
\definecolor{dialinecolor}{rgb}{0.000000, 0.000000, 0.000000}
\pgfsetstrokecolor{dialinecolor}
\draw (6.157500\du,3.985000\du)--(6.157500\du,5.485000\du);
}
\definecolor{dialinecolor}{rgb}{0.000000, 0.000000, 0.000000}
\pgfsetstrokecolor{dialinecolor}
\node[anchor=west] at (6.157500\du,3.235000\du){};
\definecolor{dialinecolor}{rgb}{0.000000, 0.000000, 0.000000}
\pgfsetstrokecolor{dialinecolor}
\node at (6.157500\du,3.235000\du){\large $s_0$};
\end{tikzpicture}}
\end{center}
\caption{The initial knowledge of an agent for the GEP-UW for $G_{2,4}$ ($G_{2,n}$ respectively).}
\label{initial}
\end{figure}

\begin{proof}

Let \textsc{alg} be an arbitrary deterministic algorithm. From the starting point $s_0$, the agent sees two edges, as presented in Figure \ref{initial}. In the hard instances we are considering, both of these edges have weight 1. In the first time step, the agent chooses to travel one of the two edges. Depending on the agent's first move, the adversary chooses one of two hard instances, $I_1$ and $I_2$. Instances $I_1$ and $I_2$ with the optimal paths for them are presented in Figure \ref{eka1}. We can see that $\textsc{cost}_{\textsc{ opt}}(G_{2,4})=8$ for both cases. Let us next analyze all the further possible movements of the agent on these two instances, $I_1$ and $I_2$.

\begin{figure}
\begin{center}
\scalebox{1.0}{\input{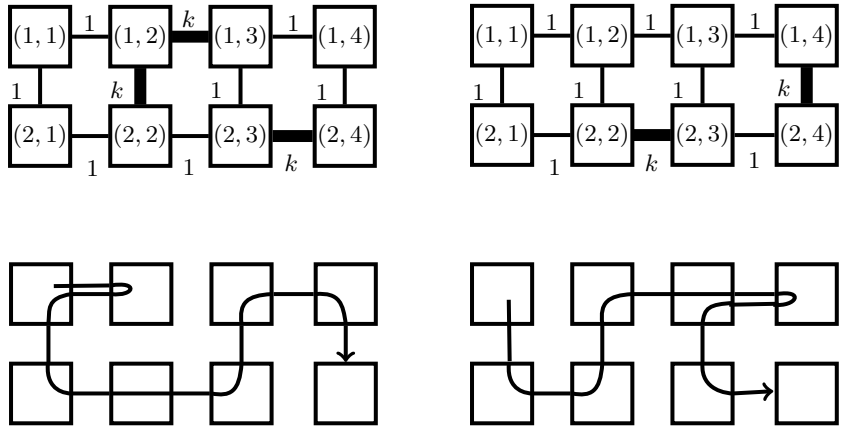}}

\end{center}
\caption{Two hard instances, $I_1$ and $I_2$, for $G_{2,4}$, and the optimal paths for them.}
\label{eka1}
\end{figure}

\begin{figure}
\begin{center}
\scalebox{1.0}{\input{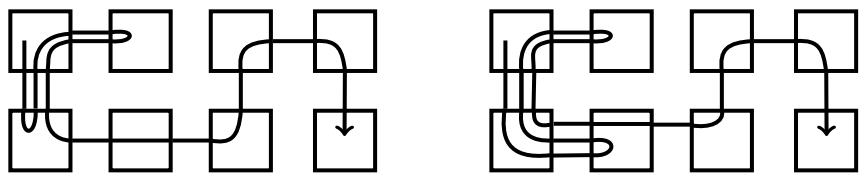}}

\end{center}
\caption{The agent's paths for the hard instance of $G_{2,4}$, in the case $I_1$.}
\label{eka2}
\end{figure}

We consider the case $I_1$, with a graph $G_{2,4}$. This case is chosen by the adversary if the agent first chooses to travel down to the vertex $(2,1)$. We can divide our analysis into two cases. Either the agent turns back to the starting vertex $s_0$, or it continues right to the vertex $(2,2)$. We have presented both of these cases in Figure \ref{eka2}. If the agent returns to $s_0$, it will naturally continue to the vertex $(1,2)$. After that, as the agent will not travel an edge with weight $k$ (this would immediately lead to a total cost of at least $7 + k-1 \geq 11$), it will turn back and travel counterclockwise until it reaches the vertex $(2,3)$. At last, the agent travels the path $\{(2,3),(1,3),(1,4),(2,4)\}$ that avoids all the edges with weight $k$. In this case, the cost of $p_A$ is 10. On the other hand, if the agent continues to the vertex $(2,2)$ after the first move, there are again two possible ways for it to proceed. Either to turn back to the vertex $(2,1)$, or continue to the vertex $(2,3)$. Let us first assume that the agent turns back. Again, the agent travels clockwise until the vertex $(1,2)$. At this point, it avoids the edge with a weight $k$, and explores the rest of the vertices exactly as it did before, that is, it takes a path $\{(1,2),(1,1),(2,1),(2,2),(2,3),(1,3),(1,4),(2,4)\}$. We see that, in this case, the cost of $p_A$ is 12. Let us now assume that, after the second time step, the agent continues to the right-hand side. Obviously, the agent must turn back to visit the vertex $(1,2)$ at some point, hence, the further right the agent continues, the worse the path becomes, leading to a cost of at least 12 in this case. Thus, the competitive ratio is $\frac{10}{8}$ in the case $I_1$.

\begin{figure}
\begin{center}
\scalebox{1.0}{\input{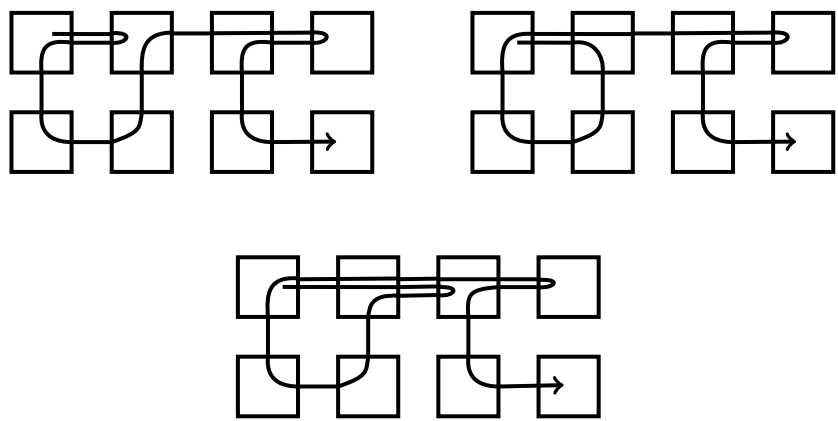}}

\end{center}
\caption{The agent's paths for the hard instances $G_{2,4}$, in the case $I_2$.}
\label{eka3}
\end{figure}

The second case, $I_2$, is presented by the adversary if the agent first chooses to travel right to the vertex $(1,2)$. This time, we divide our analysis into three cases. Either the agent turns back to the starting vertex $s_0$, or it continues right to the vertex $(1,3)$, or it goes down to the vertex $(2,2)$. We have presented all of these cases in Figure \ref{eka3}. If the agent returns to $s_0$, it will naturally continue to the vertices $(2,1)$ and $(2,2)$. After that, the agent avoids an edge with a weight $k$ (for the same reason as in the first case), and travels the path $\{(2,2),(1,2),(1,3)\}$. Next, at the vertex $(1,3)$, the agent has to decide whether to keep going right or to go down. However, it is easy to see that the cheapest way is to take a path $\{(1,3),(1,4),(1,3),(2,3),(2,4)\}$. Hence, the cost of $p_A$ in the first case of this second instance is 10. The other choice for the agent after the first move is to go down to the vertex $(2,2)$. As the agent will not travel an edge with weight $k$, it must return to the vertex $(1,2)$ after visiting the vertex $(2,1)$, either via $(2,2)$ or via $(1,1)$. After that, the agent is in a identical situation as in the first case, hence, it will travel a path $\{(1,2),(1,3),(1,4),(1,3),(2,3),(2,4)\}$. The cost of $p_A$ is again 10. In the last case, the agent keeps going right to the vertex $(1,3)$ after the first move. Because of the locations of the edges with weight $k$, it is clear that the agent must return the same way at some point to explore the vertices $(2,1)$ and $(2,2)$, hence, the further right the agent continues, the worse the algorithm becomes leading to a cost of at least 12 in this case. It follows that the competitive ratio is $\frac{10}{8}$, also in the case $I_2$.

Thus the strict competitive ratio for \textsc{alg} for $G_{2,4}$ is $\frac{10}{8} = \frac{5}{4}$. \qed

\end{proof}

Let us next generalize the result from Theorem 1 to the case of arbitrarily large $2 \times n$ grids.

\begin{theorem} For some $n \in \N$ that is a multiple of 4, let $G_{2,n}$ be a known undirected and weighted ladder graph, where weights of the edges are 1 or $k$, for some constant $k$. The competitive ratio for any deterministic online algorithm \textsc{alg} has a lower bound of $\frac{11}{9}$ for GEP-UW on $G_{2,n}$. 

\end{theorem}

\begin{proof} To prove Theorem 2, we employ the proof of Theorem 1. We note that both optimal solutions for the two instances used in the proof of Theorem 1 end in the vertex $(1,n)$, and have the same cost, 8. As, for any algorithm, visiting all the vertices and ending the vertex $(1,n)$ is a harder task than just visiting all the vertices, we know that the solution computed by any deterministic algorithm is at least as expensive as if the algorithm would end in some other vertex.

\begin{figure}
\begin{center}
\scalebox{0.73}{\input{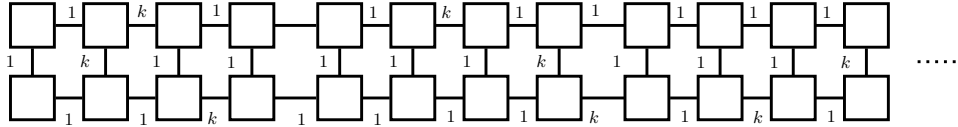}}

\end{center}
\caption{An example of hard instance of $G_{2,n}$.}
\label{eka5}
\end{figure}

We will construct the hard instance $G_{2,n}$ by using the two hard instances from the proof of Theorem 1, which are weighted graphs $G_{2,4}$, as gadgets. In $G_{2,n}$, every second gadget is a mirror image of the corresponding $G_{2,4}$ (i.\,e., the upper and lower rows are reversed ). Hence, when we have a row of such gadgets, the odd ones are identical to the original $G_{2,4}$ from Theorem 1, and the even ones are mirrored from the original $G_{2,4}$. An edge connecting an original $G_{2,4}$ and the mirror image of it, in that order, has a weight $k$, if it is between vertices $(1,x)$ and $(1,y)$, where $x,y \in \{1,2,\ldots,n \}$, and a weight 1 if it is between vertices $(2,x)$ and $(2,y)$, where $x,y \in \{1,2,\ldots,n \}$. Correspondingly, an edge connecting a mirror image and its original $G_{2,4}$, in that order, has a weight $1$, if it is between vertices $(1,x)$ and $(1,y)$, where $x,y \in \{1,2,\ldots,n \}$, and a weight k if it is between vertices $(2,x)$ and $(2,y)$, where $x,y \in \{1,2,\ldots,n \}$. We have demonstrated one example of the $G_{2,n}$ in Figure \ref{eka5}, in the picture, the edges connecting the gadgets are marked with a longer line.

Because of the gadgets, a deterministic algorithm will behave as above in the proof of Theorem 1. After moving from one gadget to the next, the algorithm is again facing the same situation. No matter which of the two edges it chooses, the adversary can present the respective hard $G_{2,4}$ as the next gadget. Obviously, any algorithm that leaves some vertices unvisited when going to the next gadget has to return again at some point which will increase its cost too much compared to any algorithm that visits the gadgets in order.

If we have $m = \frac{n}{4}$ gadgets in one instance, then the cost of the optimal solution is $m(8+1)$. The estimated cost of the solutions computed by any deterministic algorithm  must be at least $m(10+1)$. It follows that the lower bound of the competition ratio for the QEP-UW for $G_{2,n}$ is $\frac{11}{9}+ \alpha$, where $\alpha$ is a constant. \qed

\end{proof}

\subsection{Upper Bound}

Next, in Theorem 3, we complement our result with an upper bound on the competitive ratio. Thus, we prove that, there exists a deterministic online algorithm \textsc{alg}, that achieves a competitive ratio of 4 for GEP-UW on a known undirected and weighted ladder graph $G_{2,n}$. 

First, we present and prove Lemma 1, that we will use in Theorem 3. There, we estimate the cost of finishing the task, when the agent has already traveled from starting vertex $s_0$ to the vertex $(1,n)$.

\begin{lemma} Let $G_{2,n}$ be a known undirected and weighted ladder graph, where weights of the edges are 1 or $k$, where $k$ is a constant. If the agent has already traveled from $s_0$ to the vertex $(1,n)$, it can visit the rest of the vertices via a path whose cost is asymptotically at most twice that of an optimal GEP-UW solution, for growing values of k.

\end{lemma}

\begin{proof} The agent has already traveled from the vertex $(1,1)$ to the vertex $(1,n)$. Thus, it knows the weight of at least $\frac{2}{3}$ of all the edges. Still, at most $\frac{1}{3}$ of all the edge weights are unknown for the agent. We first assume that the agent traveled from the vertex $(1,1)$ to the vertex $(1,n)$ via the path $(1,1),(1,2),\ldots,(1,n)$. Later on, we will generalize our argument to any path from $(1,1)$ to $(1,n)$. We consider the problem by dividing the proof into two cases, because, if, on the way back, the agent stands above the rightmost unvisited vertex, there are two possibilities.

In the first case, we assume, that the agent faces the situation, where there are $m$ visited vertices on its left side, that are connected to each other with edges that have a weight 1, with any $1 <m\leq k$. Also, from the first $m-1$ vertices the edges leading to the lower row have a weight $k$, but the edge leading to the lower row from the $m$th vertex has a weight 1. With $m=1$, this naturally includes also the situation where there are no $k$-edges leading to the lower row. We have demonstrated these situations in Figure \ref{lemma1.1}, with $m=4$ and $m=1$. Here, the agent first makes a loop, by traveling first to the left on the upper row to explore again the $m$ vertices and then traveling to the right on the lower row, as shown in Figure \ref{lemma1.2}. After that, the agent moves to the left on a cheapest path to the next unvisited vertex or a vertex above it. During the loop, if the subpath consisting of unvisited and visited or only unvisited edges contains an edge of weight $k$, then also the optimal path must travel an edge of weight $k$ due to the ``$k$-separator'' we have presented in Figure \ref{lemma1.3}. The agent's way back to the left travels through a completely known part of the graph and is thus optimal. It follows that the cost of the agent's path in this case is at most two times the optimal cost. 

\begin{figure}
\begin{center}
\scalebox{0.9}{\input{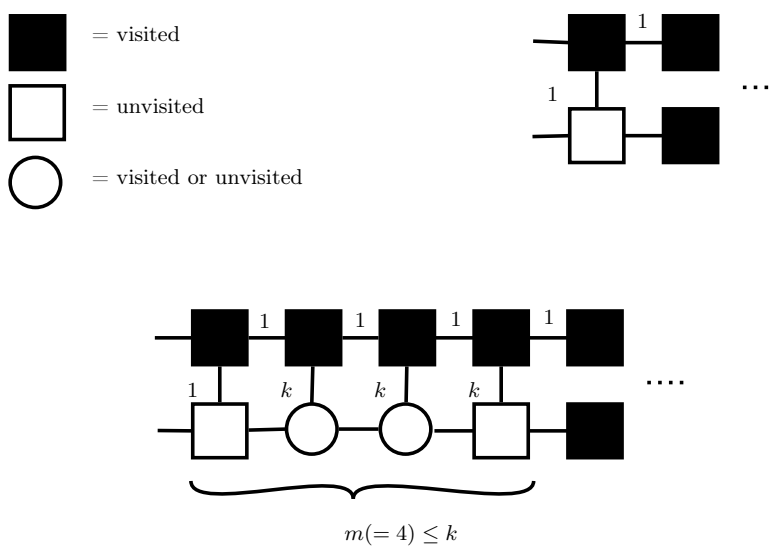}}

\end{center}
\caption{Examples for the case 1, with $m=1$, on the upper picture, and $m=4$, on the lower picture.}
\label{lemma1.1}
\end{figure}

In the second case, we assume that no path as in first case exists. It follows, that the agent does not travel back left, but takes the edge leading to the lower row, no matter the cost of this edge. We have demonstrated this situation in Figure \ref{lemma1.4}. Here, after traveling to the lower row, the agent travels left on a cheapest path to the next unvisited vertex or the one above it. Thus, the agent is optimal in this case. 

In the general case, i.\,e., if the path of the agent from the vertex $(1,1)$ to the vertex $(1,n)$ includes some vertices from the bottom row, we can extend our argumentation as follows: For columns that consist of visited vertices only, the agent knows the cost of all incident edges and can just traverse these columns in an optimal way. For paths of the ladder, where the agent only used the bottom row on its way from the vertex $(1,1)$ to the vertex $(1,n)$, we can use the same argument, just with exchanged roles of the top and bottom row. Note that these parts are always separated by some completely visited columns.

We conclude, that \textsc{alg} is at least 2-competitive on this problem. \qed

\begin{figure}
\begin{center}
\scalebox{0.9}{\input{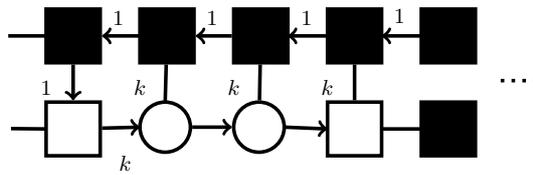}}

\end{center}
\caption{The agent's path on the first case, with $m=4$.}
\label{lemma1.2}
\end{figure}

 \begin{figure}
\begin{center}
\scalebox{0.9}{\input{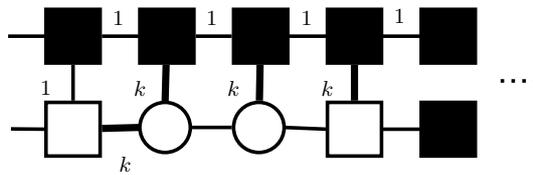}}

\end{center}
\caption{The situation where the edges with a weight $k$ form a ``$k$-separator''.}
\label{lemma1.3}
\end{figure}

\end{proof}

In the following theorem, we prove an upper bound of 4 for the competitive ratio on $2 \times n$ grids.

\begin{theorem} Let $G_{2,n}$ be a known undirected and weighted ladder graph, where weights of the edges are 1 or $k$, where $k$ is an constant. There exists a 4-competitive deterministic online algorithm \textsc{alg} for GEP-UW on $G_{2,n}$.

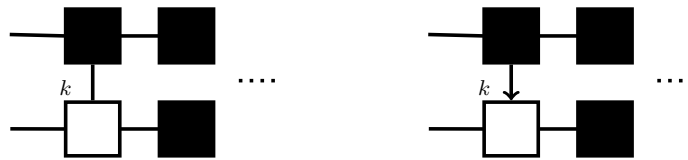
\begin{figure}
\begin{center}
\scalebox{0.9}{
\ifx\du\undefined
  \newlength{\du}
\fi
\setlength{\du}{15\unitlength}
\begin{tikzpicture}
\pgftransformxscale{1.000000}
\pgftransformyscale{-1.000000}
\definecolor{dialinecolor}{rgb}{0.000000, 0.000000, 0.000000}
\pgfsetstrokecolor{dialinecolor}
\definecolor{dialinecolor}{rgb}{1.000000, 1.000000, 1.000000}
\pgfsetfillcolor{dialinecolor}
\pgfsetlinewidth{0.100000\du}
\pgfsetdash{}{0pt}
\pgfsetdash{}{0pt}
\pgfsetmiterjoin
\definecolor{dialinecolor}{rgb}{0.000000, 0.000000, 0.000000}
\pgfsetfillcolor{dialinecolor}
\fill (8.400000\du,4.000000\du)--(8.400000\du,5.500000\du)--(9.900000\du,5.500000\du)--(9.900000\du,4.000000\du)--cycle;
\definecolor{dialinecolor}{rgb}{0.000000, 0.000000, 0.000000}
\pgfsetstrokecolor{dialinecolor}
\draw (8.400000\du,4.000000\du)--(8.400000\du,5.500000\du)--(9.900000\du,5.500000\du)--(9.900000\du,4.000000\du)--cycle;
\pgfsetlinewidth{0.100000\du}
\pgfsetdash{}{0pt}
\pgfsetdash{}{0pt}
\pgfsetmiterjoin
\definecolor{dialinecolor}{rgb}{0.000000, 0.000000, 0.000000}
\pgfsetfillcolor{dialinecolor}
\fill (11.000000\du,4.000000\du)--(11.000000\du,5.500000\du)--(12.500000\du,5.500000\du)--(12.500000\du,4.000000\du)--cycle;
\definecolor{dialinecolor}{rgb}{0.000000, 0.000000, 0.000000}
\pgfsetstrokecolor{dialinecolor}
\draw (11.000000\du,4.000000\du)--(11.000000\du,5.500000\du)--(12.500000\du,5.500000\du)--(12.500000\du,4.000000\du)--cycle;
\pgfsetlinewidth{0.100000\du}
\pgfsetdash{}{0pt}
\pgfsetdash{}{0pt}
\pgfsetbuttcap
{
\definecolor{dialinecolor}{rgb}{0.000000, 0.000000, 0.000000}
\pgfsetfillcolor{dialinecolor}
\definecolor{dialinecolor}{rgb}{0.000000, 0.000000, 0.000000}
\pgfsetstrokecolor{dialinecolor}
\draw (9.900000\du,4.750000\du)--(11.000000\du,4.750000\du);
}
\pgfsetlinewidth{0.100000\du}
\pgfsetdash{}{0pt}
\pgfsetdash{}{0pt}
\pgfsetmiterjoin
\definecolor{dialinecolor}{rgb}{1.000000, 1.000000, 1.000000}
\pgfsetfillcolor{dialinecolor}
\fill (8.400000\du,6.600000\du)--(8.400000\du,8.100000\du)--(9.900000\du,8.100000\du)--(9.900000\du,6.600000\du)--cycle;
\definecolor{dialinecolor}{rgb}{0.000000, 0.000000, 0.000000}
\pgfsetstrokecolor{dialinecolor}
\draw (8.400000\du,6.600000\du)--(8.400000\du,8.100000\du)--(9.900000\du,8.100000\du)--(9.900000\du,6.600000\du)--cycle;
\pgfsetlinewidth{0.100000\du}
\pgfsetdash{}{0pt}
\pgfsetdash{}{0pt}
\pgfsetmiterjoin
\definecolor{dialinecolor}{rgb}{0.000000, 0.000000, 0.000000}
\pgfsetfillcolor{dialinecolor}
\fill (11.000000\du,6.600000\du)--(11.000000\du,8.100000\du)--(12.500000\du,8.100000\du)--(12.500000\du,6.600000\du)--cycle;
\definecolor{dialinecolor}{rgb}{0.000000, 0.000000, 0.000000}
\pgfsetstrokecolor{dialinecolor}
\draw (11.000000\du,6.600000\du)--(11.000000\du,8.100000\du)--(12.500000\du,8.100000\du)--(12.500000\du,6.600000\du)--cycle;
\pgfsetlinewidth{0.100000\du}
\pgfsetdash{}{0pt}
\pgfsetdash{}{0pt}
\pgfsetbuttcap
{
\definecolor{dialinecolor}{rgb}{0.000000, 0.000000, 0.000000}
\pgfsetfillcolor{dialinecolor}
\definecolor{dialinecolor}{rgb}{0.000000, 0.000000, 0.000000}
\pgfsetstrokecolor{dialinecolor}
\draw (9.900000\du,7.350000\du)--(11.000000\du,7.350000\du);
}
\definecolor{dialinecolor}{rgb}{0.000000, 0.000000, 0.000000}
\pgfsetstrokecolor{dialinecolor}
\node[anchor=west] at (9.150000\du,7.350000\du){};
\pgfsetlinewidth{0.100000\du}
\pgfsetdash{}{0pt}
\pgfsetdash{}{0pt}
\pgfsetbuttcap
{
\definecolor{dialinecolor}{rgb}{0.000000, 0.000000, 0.000000}
\pgfsetfillcolor{dialinecolor}
\definecolor{dialinecolor}{rgb}{0.000000, 0.000000, 0.000000}
\pgfsetstrokecolor{dialinecolor}
\draw (9.150000\du,6.549878\du)--(9.150000\du,5.550122\du);
}
\definecolor{dialinecolor}{rgb}{0.000000, 0.000000, 0.000000}
\pgfsetstrokecolor{dialinecolor}
\node at (8.400000\du,6.200000\du){$k$};
\pgfsetlinewidth{0.100000\du}
\pgfsetdash{{\pgflinewidth}{0.200000\du}}{0cm}
\pgfsetdash{{\pgflinewidth}{0.200000\du}}{0cm}
\pgfsetbuttcap
{
\definecolor{dialinecolor}{rgb}{0.000000, 0.000000, 0.000000}
\pgfsetfillcolor{dialinecolor}
\definecolor{dialinecolor}{rgb}{0.000000, 0.000000, 0.000000}
\pgfsetstrokecolor{dialinecolor}
\draw (13.200000\du,6.000000\du)--(14.250600\du,6.000000\du);
}
\pgfsetlinewidth{0.100000\du}
\pgfsetdash{}{0pt}
\pgfsetdash{}{0pt}
\pgfsetbuttcap
{
\definecolor{dialinecolor}{rgb}{0.000000, 0.000000, 0.000000}
\pgfsetfillcolor{dialinecolor}
\definecolor{dialinecolor}{rgb}{0.000000, 0.000000, 0.000000}
\pgfsetstrokecolor{dialinecolor}
\draw (6.850000\du,4.710000\du)--(8.400000\du,4.750000\du);
}
\pgfsetlinewidth{0.100000\du}
\pgfsetdash{}{0pt}
\pgfsetdash{}{0pt}
\pgfsetbuttcap
{
\definecolor{dialinecolor}{rgb}{0.000000, 0.000000, 0.000000}
\pgfsetfillcolor{dialinecolor}
\definecolor{dialinecolor}{rgb}{0.000000, 0.000000, 0.000000}
\pgfsetstrokecolor{dialinecolor}
\draw (6.866273\du,7.346273\du)--(8.400000\du,7.350000\du);
}
\pgfsetlinewidth{0.100000\du}
\pgfsetdash{}{0pt}
\pgfsetdash{}{0pt}
\pgfsetmiterjoin
\definecolor{dialinecolor}{rgb}{0.000000, 0.000000, 0.000000}
\pgfsetfillcolor{dialinecolor}
\fill (20.000000\du,4.000000\du)--(20.000000\du,5.500000\du)--(21.500000\du,5.500000\du)--(21.500000\du,4.000000\du)--cycle;
\definecolor{dialinecolor}{rgb}{0.000000, 0.000000, 0.000000}
\pgfsetstrokecolor{dialinecolor}
\draw (20.000000\du,4.000000\du)--(20.000000\du,5.500000\du)--(21.500000\du,5.500000\du)--(21.500000\du,4.000000\du)--cycle;
\pgfsetlinewidth{0.100000\du}
\pgfsetdash{}{0pt}
\pgfsetdash{}{0pt}
\pgfsetmiterjoin
\definecolor{dialinecolor}{rgb}{0.000000, 0.000000, 0.000000}
\pgfsetfillcolor{dialinecolor}
\fill (22.600000\du,4.000000\du)--(22.600000\du,5.500000\du)--(24.100000\du,5.500000\du)--(24.100000\du,4.000000\du)--cycle;
\definecolor{dialinecolor}{rgb}{0.000000, 0.000000, 0.000000}
\pgfsetstrokecolor{dialinecolor}
\draw (22.600000\du,4.000000\du)--(22.600000\du,5.500000\du)--(24.100000\du,5.500000\du)--(24.100000\du,4.000000\du)--cycle;
\pgfsetlinewidth{0.100000\du}
\pgfsetdash{}{0pt}
\pgfsetdash{}{0pt}
\pgfsetbuttcap
{
\definecolor{dialinecolor}{rgb}{0.000000, 0.000000, 0.000000}
\pgfsetfillcolor{dialinecolor}
\definecolor{dialinecolor}{rgb}{0.000000, 0.000000, 0.000000}
\pgfsetstrokecolor{dialinecolor}
\draw (21.500000\du,4.750000\du)--(22.600000\du,4.750000\du);
}
\pgfsetlinewidth{0.100000\du}
\pgfsetdash{}{0pt}
\pgfsetdash{}{0pt}
\pgfsetmiterjoin
\definecolor{dialinecolor}{rgb}{1.000000, 1.000000, 1.000000}
\pgfsetfillcolor{dialinecolor}
\fill (20.000000\du,6.600000\du)--(20.000000\du,8.100000\du)--(21.500000\du,8.100000\du)--(21.500000\du,6.600000\du)--cycle;
\definecolor{dialinecolor}{rgb}{0.000000, 0.000000, 0.000000}
\pgfsetstrokecolor{dialinecolor}
\draw (20.000000\du,6.600000\du)--(20.000000\du,8.100000\du)--(21.500000\du,8.100000\du)--(21.500000\du,6.600000\du)--cycle;
\pgfsetlinewidth{0.100000\du}
\pgfsetdash{}{0pt}
\pgfsetdash{}{0pt}
\pgfsetmiterjoin
\definecolor{dialinecolor}{rgb}{0.000000, 0.000000, 0.000000}
\pgfsetfillcolor{dialinecolor}
\fill (22.600000\du,6.600000\du)--(22.600000\du,8.100000\du)--(24.100000\du,8.100000\du)--(24.100000\du,6.600000\du)--cycle;
\definecolor{dialinecolor}{rgb}{0.000000, 0.000000, 0.000000}
\pgfsetstrokecolor{dialinecolor}
\draw (22.600000\du,6.600000\du)--(22.600000\du,8.100000\du)--(24.100000\du,8.100000\du)--(24.100000\du,6.600000\du)--cycle;
\pgfsetlinewidth{0.100000\du}
\pgfsetdash{}{0pt}
\pgfsetdash{}{0pt}
\pgfsetbuttcap
{
\definecolor{dialinecolor}{rgb}{0.000000, 0.000000, 0.000000}
\pgfsetfillcolor{dialinecolor}
\definecolor{dialinecolor}{rgb}{0.000000, 0.000000, 0.000000}
\pgfsetstrokecolor{dialinecolor}
\draw (21.500000\du,7.350000\du)--(22.600000\du,7.350000\du);
}
\definecolor{dialinecolor}{rgb}{0.000000, 0.000000, 0.000000}
\pgfsetstrokecolor{dialinecolor}
\node[anchor=west] at (20.750000\du,7.350000\du){};
\pgfsetlinewidth{0.100000\du}
\pgfsetdash{}{0pt}
\pgfsetdash{}{0pt}
\pgfsetbuttcap
{
\definecolor{dialinecolor}{rgb}{0.000000, 0.000000, 0.000000}
\pgfsetfillcolor{dialinecolor}
\pgfsetarrowsstart{to}
\definecolor{dialinecolor}{rgb}{0.000000, 0.000000, 0.000000}
\pgfsetstrokecolor{dialinecolor}
\draw (20.750000\du,6.549878\du)--(20.750000\du,5.550122\du);
}
\definecolor{dialinecolor}{rgb}{0.000000, 0.000000, 0.000000}
\pgfsetstrokecolor{dialinecolor}
\node at (20.000000\du,6.200000\du){$k$};
\pgfsetlinewidth{0.100000\du}
\pgfsetdash{{\pgflinewidth}{0.200000\du}}{0cm}
\pgfsetdash{{\pgflinewidth}{0.200000\du}}{0cm}
\pgfsetbuttcap
{
\definecolor{dialinecolor}{rgb}{0.000000, 0.000000, 0.000000}
\pgfsetfillcolor{dialinecolor}
\definecolor{dialinecolor}{rgb}{0.000000, 0.000000, 0.000000}
\pgfsetstrokecolor{dialinecolor}
\draw (24.800000\du,6.000000\du)--(25.850600\du,6.000000\du);
}
\pgfsetlinewidth{0.100000\du}
\pgfsetdash{}{0pt}
\pgfsetdash{}{0pt}
\pgfsetbuttcap
{
\definecolor{dialinecolor}{rgb}{0.000000, 0.000000, 0.000000}
\pgfsetfillcolor{dialinecolor}
\definecolor{dialinecolor}{rgb}{0.000000, 0.000000, 0.000000}
\pgfsetstrokecolor{dialinecolor}
\draw (18.450000\du,4.710000\du)--(20.000000\du,4.750000\du);
}
\pgfsetlinewidth{0.100000\du}
\pgfsetdash{}{0pt}
\pgfsetdash{}{0pt}
\pgfsetbuttcap
{
\definecolor{dialinecolor}{rgb}{0.000000, 0.000000, 0.000000}
\pgfsetfillcolor{dialinecolor}
\definecolor{dialinecolor}{rgb}{0.000000, 0.000000, 0.000000}
\pgfsetstrokecolor{dialinecolor}
\draw (18.466273\du,7.346273\du)--(20.000000\du,7.350000\du);
}
\end{tikzpicture}}

\end{center}
\caption{Case 2, in the left picture, and the agent's path for it, in the right picture.}
\label{lemma1.4}
\end{figure}

\end{theorem}

\begin{proof}

We have a graph $G_{2,n}$, and we consider the following deterministic algorithm \textsc{alg}. \textsc{alg} works for GEP-UW on $G_{2,n}$. As we are considering ladders, we can employ the strategy where the agent always goes from left to right as long as it is possible to use only the edges with a weight 1. If the algorithm is forced to travel an edge with a weight $k$, it turns back to visit all the possible vertices it has not explored so far. From the proof of Lemma 1, we know that, for any subgraph of $G_{2,n}$, it is possible to visit the unvisited vertices from the right end of the subgraph to the starting vertex of the subgraph with a cost that is at most twice the cost of the optimal traversal of this subgraph. We consider two different situations, whether the optimal path does or does not include any edges with weight $k$.

If the optimal path does not include any edge with a weight $k$, we can be sure that neither does the one traveled by the agent. This follows because we know that there exists a path that consist of only edges with weight 1, and we know that the agent avoids the edges with weight $k$ if possible. The agent realizes whether it missed some vertices on its way from left to right. Then, when returning, it can visit those vertices. It follows that the agent uses some path to reach the righthand end of the grid and then, on its way back, follows a path as described in Lemma 1. As the agent's first path must naturally be shorter than the optimal path, we can conclude that $\textsc{comp}_{\textsc{alg}}\leqslant 2$ in this case.

We now consider the other case, i.\,e., that the optimal solution uses an expensive edge, and we divide it into the following cases.

\begin{figure}
\begin{center}
\scalebox{0.6}{\input{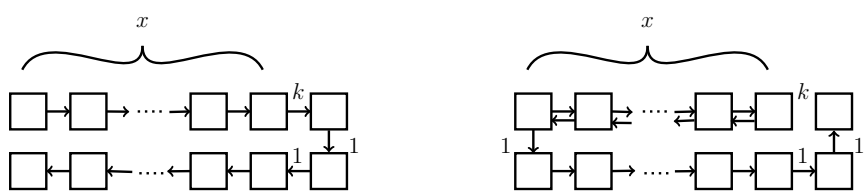}}

\end{center}
\caption{The optimal path and the agent's path in case 1.}
\label{case3}
\end{figure}

\textbf{Case 1.} Suppose there is a dead-end in the graph, from which the only exits are either a long way back or a $k$-edge, such that the optimal path traverses this $k$-edge. However, assume that $\textsc{alg}$ goes back. Let us assume that this long way back has cost $x$. It follows that the optimal solution pays at least $x+k+2+x$ for visiting all vertices in these $x$ columns of the grid and leaving via the $k$-edge, and the agent pays at most $3x+3$ for traversing the columns and leaving to the righthand side. We have presented an example of the optimal path and the agent's path in Figure \ref{case3}.

\begin{figure}
\begin{center}
\scalebox{0.6}{\input{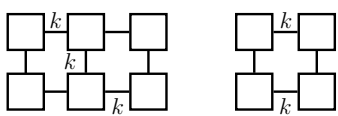}}

\end{center}
\caption{Example forms of the $k$-separators.}
\label{situations}
\end{figure}

\textbf{Case 2.} Suppose, that the grid is separated by $k$-edges. Then, we have one of the following two situations. In the first type of separator, edges between $(i,j)$ and $(i+1,j)$, $(i',j')$ and $(i'+1,j')$, where $i'>i$ and $j\neq j'$, and for all $\ell \in \{i+1,\ldots,i'\}$ edges between $(\ell,1)$ and $ (\ell,2)$ are $k$-edges. In the second type of separator, edges between $(i,j)$ and $(i+1,j)$, $(i,j')$ and $(i+1,j')$, where $j+j'=3$ are $k$-edges. Examples of both cases are demonstrated in Figure \ref{situations}. 

In these situations, the agent proceeds as follows. The agent first identifies the $k$-edge-separator. After that, the agent goes back and explores the missing vertices. By Lemma 1, we know that this is possible. Next, the agent goes to the $k$-edge-separator via the cheapest way, and then traverses one of the $k$-edges. At last, the agent moves to the right and starts with the first step. The agent's way of proceeding will not change, even if there are more than one $k$-separators in a row. In that case, it will find an almost optimal path through all the $k$-edges, as it always identifies the $k$-edge-separator before traversing one edge of it.

We can estimate that the cost of $p_A$ is at most four times the cost of $p_O$ by separately counting each component that does not contain a $k$-edge-separator. That is, a way to the separator, a 2-competitive way back to explore everything, and again the way to the separator. It is important to notice that each expensive $k$-separator is only used once. Thus, $\textsc{comp}_{\textsc{alg}}\leqslant 4$. \qed

\end{proof}

\section{Directed Graphs}

In this section, we consider the case of directed graphs. Here, we consider the natural greedy algorithm \textsc{gr}, which always makes a best possible decision at the moment, that is, it chooses the shortest way to some new, unvisited vertex. In the following example, we present directed weighted graphs $G_{2,n}$ and $G_{n,n}$, with weights $1$ and $k$, such that $\textsc{comp}_{\textsc{GR}}(G)\rightarrow  \frac{n}{a},$ where $G \in \{G_{2,n},G_{n,n}\}$, $a$ is a constant, and $n$ tends to infinity.

\begin{example}

Consider a directed grid $G_{2,n}$ as in Figure \ref{katos}. We claim that the algorithm \textsc{gr} is a very bad strategy for GEP-DW. 

\begin{figure}
\begin{center}
\scalebox{0.83}{\input{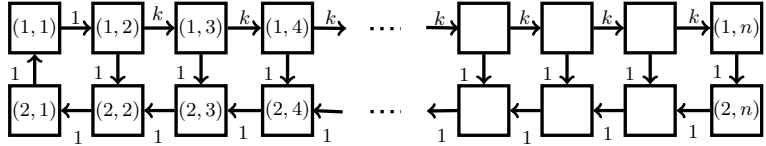}}
\end{center}

\caption{The directed grid $G_{2,n}$, for which the \textsc{gr} is a very bad strategy.}
\label{katos}
\end{figure}

We prove the claim by simply explaining the worst-case scenario of the algorithm's execution. That is, \textsc{gr} makes $n-2$ loops from the starting point and in the end travels once the upper row from left to right and finally from vertex $(1,n)$ to the vertex $(2,n)$. However, the optimal path simply goes around the graph. Here, the competitive ratio depends only the amount of traveled $k$-edges. The agent travels some $k$-edges on each of the loops after the first loop, and while traveling once the upper row from left to right. The total amount of $k$-edges the agent travels is $\sum_{i}^{n-2} i = \frac{(n-2)(n-1)}{2} = \frac{n^2-3n+2}{2}$. As the optimal path is just a circle, the amount of the $k$-edges is $n-2$. It follows that the competitive ratio is $$\textsc{comp}_{\textsc{gr}}(G_{2,n})=\frac{n^2-3n+2}{2(n-2)},$$ hence,   $$\textsc{comp}_{\textsc{gr}}(G_{2,n}) \xrightarrow [n \rightarrow \infty]{} \frac{n}{2}.$$

Next, we enlarge our graph $G_{2,n}$, with slight modifications, to the
$n\times n $ grid, that is presented with $n=5$ in Figure \ref{vika1}. We
prove that \textsc{gr} is again a very bad strategy for GEP-DW. Similarly
as with the graph $G_{2,n}$, the agent makes a lot of loops whereas the
optimal path simply goes around. In Figure \ref{vika2}, we have presented
$p_A$ and $p_O$ for the graph $G_{5,5}$.
Again, it is clear that the competitive ratio depends only the amount of
traveled $k$-edges. It is easy to see that the optimal path travels twice
each of the horizontal $k$-edges, and once each of the vertical $k$-edges.
Hence, the total cost of traveled $k$-edges for the optimal solution, for
$G_{5,5}$, is $2(2(n-3))+2=4n-10$. The agent travels some $k$-edges on
$n-2$ loops, and while traveling once the upper row from left to right, on
the first two rows, and on $n-1$ loops, on the rows three and four, and
once each of the $k$-edge on the columns. It follows that, the total cost
of traveled $k$-edges for the agent's solution, for $G_{5,5}$, is
$\sum_i^{n-2}i + \sum_j^{n-1}j +2=\frac{2n^2-4n+6}{2}$. Thus,
$$\textsc{comp}_{\textsc{gr}}(G_{5,5})=\frac{2n^2-4n+6}{2(4n-10)}.$$

Let us now consider any grid $G_{n,n}$, built exactly as the grid $G_{5,5}$. With any $n$, $p_A$ and $p_0$ behave in the first 5 rows exactly as they did with the graph $G_{5,5}$. It follows that, $$\textsc{comp}_{\textsc{gr}}(G_{n,n})=\frac{(\frac{n}{5})(2n^2-4n+6)}{(\frac{n}{5})(2(4n-10))},$$ hence 
 $$\textsc{comp}_{\textsc{gr}}(G_{n,n}) \xrightarrow [n \rightarrow \infty]{} \frac{n}{4}.$$

From the example above, we can conclude that the greedy algorithm for directed grids $G_{2,n}$ and $G_{n,n}$ is a very bad strategy for GEP-DW. It follows that we must consider some other online algorithm to get a better than linear limit value of the competitive ratio $n$.

\begin{figure}
\begin{center}
\scalebox{0.7}{\input{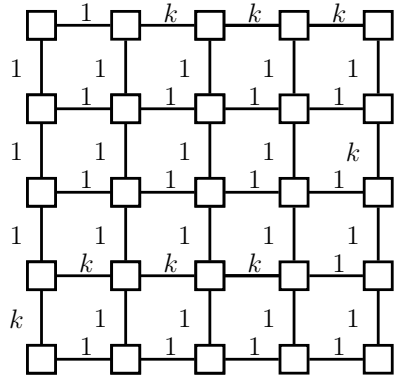}}
\end{center}

\caption{The grid $G_{5,5}$.}
\label{vika1}
\end{figure}

\begin{figure}
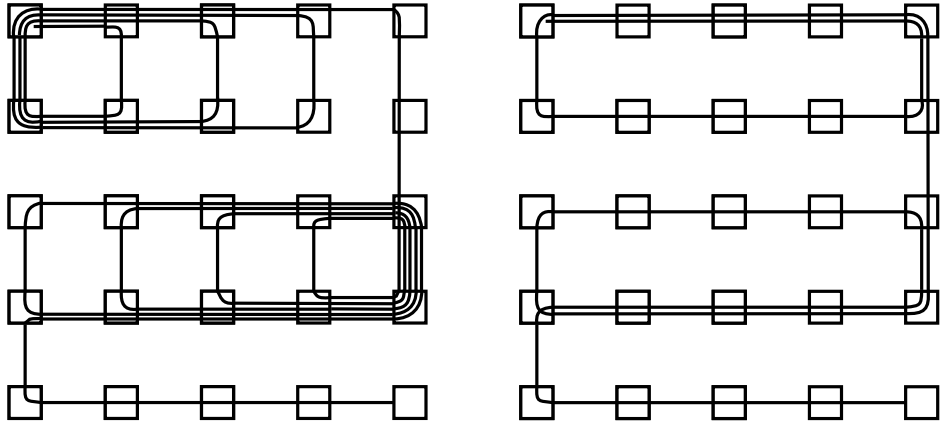

\begin{center}
\scalebox{0.8}{\input{vika2.tex}}\qquad
\scalebox{0.8}{\input{vika3.tex}}
\end{center}

\caption{The agent's path (left) and the optimal path (right) through the graph $G_{5,5}$.}
\label{vika2}
\end{figure}



\end{example}

\newpage
\addcontentsline{toc}{section}{References}    


\begin{thebibliography}{99}


\bibitem{B} 
  Piotr Berman: 
  On-line searching and navigation. 
  In: A.~Fiat, G.~Woeginger (eds.): \emph{ Online Algorithms,
    The State of the Art}, volume 1442 of \emph{Lecture Notes in Computer Science},
  pp.~232--241. Springer-Verlag, 1996. 

\bibitem{BE} 
  Allan Borodin and Ran El-Yaniv: 
  \emph{Online Computation and  Competitive Analysis}. Cambridge University Press, 1998.

\bibitem{monta} 
  Prosenjit Bose, Andrej Brodnik, Svante Carlsson, Erik D.~Demaine, Rudolf
  Fleischer, Alejandro López-Ortiz, Pat Morin, J.~Ian Munro:
  \emph{Online Routing in Convex Subdivisions}. In \emph{Proc.~of
    the 6th Italian Conference on Algorithms and Complexity (CIAC 2006)},
  volume 1969 of \emph{Lecture Notes in Computer Science},
  pp.~47--59. Springer-Verlag, 2006.

\bibitem{PK} 
  Prosenjit Bose and Pat Morin: 
  Competitive online routing in geometric graphs. \emph{Theoretical
    Computer Science} 324(2-3):273--288, 2004.

\bibitem{DP} 
  Xiaotie Deng and Christos H.~Papadimitriou. 
  Exploring an unknown graph (extended abstract).
  In \emph{Proc.~of the 31st Annual IEEE Symposium on Foundations of
    Computer Science (FOCS 1990)}, volume 1, pp.~355--361. IEEE Computer
  Society, 1990. 

\bibitem{FT} 
  Rudolf Fleischer, Gerhard Trippen:
  Exploring an Unknown Graph Efficiently. In \emph{Proc.~of the 13th European
    Symposium on Algorithms (ESA 2005)}, volume 3669 of \emph{Lecture Notes
    in Computer Science}, pp.~11--22. Springer-Verlag, 2005.

\bibitem{GK} 
  Subir K.~Ghosh and Rolf Klein: 
  Online algorithms for searching and exploration in the plane.
  \emph{Computer Science Review} 4(4):189--201, 2010.

\bibitem{KP} 
  Bala Kalyanasundaram and Kirk R. Pruhs: 
  Constructing competitive tours from local information. 
  \emph{Theoretical Computer Science} 130(1):125--138, 1994.

\bibitem{DK} 
  Dennis Komm: 
  \emph{An Introduction to Online Computation -- Determinism,
    Randomization, Advice}. Springer-Verlag 2016. 

\bibitem{MMS} 
  Nicole Megow, Kurt Mehlhorn, Pascal Schweitzer: 
  Online graph exploration: new results on old and new algorithms. 
  In \emph{Proc.~of the  38th International Conference on Automata, Languages
    and Programming (ICALP 2012)}, volume 6756 of \emph{Lecture Notes in
    Computer Science}, pp.~478--489. Springer-Verlag, 2012. 

\bibitem{ST} 
  Daniel Sleator and Robert E.~Tarjan: 
  Amortized efficiency of list update and paging rules.  
  \emph{Communications of the ACM} 28(2):202--208, 1985.

\end{thebibliography}
\end{document}